\documentclass[english,3p,twocolumn,superscriptaddress,prl]{revtex4-2}
\usepackage[toc,page]{appendix}
\usepackage[latin9]{inputenc}
\usepackage{amsmath}
\usepackage{amssymb, enumerate}
\usepackage{color}
\usepackage{babel}
\usepackage{graphicx}
\usepackage{epstopdf}
\usepackage{float}
\usepackage{thmtools}
\usepackage{thm-restate}
\usepackage{hyperref}
\usepackage{cleveref}
\usepackage{enumerate}
\usepackage{lipsum}
\usepackage{fixmath}
\usepackage{dsfont}
\usepackage{tikz}
\usepackage{pifont}
\usepackage{amsmath,amsfonts,amsthm}
\usepackage{graphicx}

\usepackage{pifont}
\usetikzlibrary{arrows.meta,bending,matrix,positioning}
\makeatletter

\@ifundefined{textcolor}{}
{%
 \definecolor{BLACK}{gray}{0}
 \definecolor{WHITE}{gray}{1}
 \definecolor{RED}{rgb}{1,0,0}
 \definecolor{GREEN}{rgb}{0,1,0}
 \definecolor{BLUE}{rgb}{0,0,1}
 \definecolor{CYAN}{cmyk}{1,0,0,0}
 \definecolor{MAGENTA}{cmyk}{0,1,0,0}
 \definecolor{YELLOW}{cmyk}{0,0,1,0}
}

\theoremstyle{theorem}
\newtheorem{theorem}{Theorem}

\newtheorem{lem}[theorem]{Lemma}
\theoremstyle{definition}
\newtheorem{definition}[theorem]{Definition}
\newtheorem{example}[theorem]{Example}
\theoremstyle{remark}

\makeatother

\begin{document}

\title{Multipartite channel assemblages}
\author{Micha{\l} Banacki}
\affiliation{International Centre for Theory of Quantum Technologies, University of Gda\'{n}sk, Jana Ba\.{z}y\'{n}skiego 1A, 80-309 Gda\'{n}sk, Poland}
\affiliation{Institute of Theoretical Physics and Astrophysics, Faculty of Mathematics, Physics and Informatics, University of Gda\'{n}sk, Wita Stwosza 57, 80-308 Gda\'{n}sk, Poland}
\author{Ravishankar Ramanathan}
\affiliation{Department of Computer Science, The University of Hong Kong, Pokfulam Road, Hong Kong}
\author{Pawe{\l} Horodecki}
\affiliation{International Centre for Theory of Quantum Technologies, University of Gda\'{n}sk, Jana Ba\.{z}y\'{n}skiego 1A, 80-309 Gda\'{n}sk, Poland}
\affiliation{Faculty of Applied Physics and Mathematics, Gda\'{n}sk University of Technology, Gabriela Narutowicza 11/12, 80-233 Gda\'{n}sk, Poland} 
%%%%%%%%%%%%%%%%%%%%%

\begin{abstract}Motivated by the recent studies on post-quantum steering, we generalize notion of bipartite channels steering by introducing the concept of multipartite no-signaling channel assemblages. We show that beyond the bipartite case, no-signaling and quantum description of such scenarios do not coincide. With a help of Choi-Jamio{\l}kowski isomorphism we present full description of considered classes of assemblages and in particular, we use this characterization to provide sufficient conditions for extremality of quantum channel assemblages in the set of all no-signaling channel assemblages. Finally, in the tripartite case, we introduce and discuss a relaxed version of channel steering where only certain subsystems obey no-signaling constrains. In this asymmetric scenario we are able to provide exactly 1 bit of key that is secure against no-signaling eavesdropper.
\end{abstract}

%%%%%%%%%%%%%%%%%%%%%


\maketitle

\textit{Introduction.-} Quantum mechanics provides a fundamental framework governing physical processes at the microscopic level. The central idea of entanglement presented within this framework, enable us to observe and utilize phenomena contradicting our macroscopic intuitions \cite{EPR} - starting from presence of correlations stronger that classically predicted \cite{Bell64} and finishing at possibility of steering local states of a distant parties \cite{S36,WJD07}.

In order to better understand the quantum paradigm from the informational perspective and relate it to the no-signaling principle of special relativity, it is often convenient to consider some generalized post-quantum theories. In particular, this approach gave a birth to considerations regarding steering scenarios with less restricted constrains \cite{HS18,SAPHS18,SBCSV15, SHSA20} and analysis of underlying convex structures with potential cryptographic applications \cite{BRH21,BMRH21,RBRH20, RHB21}.

In this paper, motivated by the recent interest in notions of multipartite and post-quantum steering, we generalized idea of (bipartite) channel steering \cite{Piani15}. We start by recalling basic definitions of no-signaling and quantum assemblages. Next, based on the above description, we introduce a notions of no-signaling, quantum and local channel assemblages and provide their characterization via related Choi matrices. With this background we discuss problem of quantum realization of extremality and related issue of security against eavesdropper holding post-quantum resources. Finally, we also introduce and discuss slightly different paradigm of post-quantum channel steering base on the idea of  no-signaling conditions restricting only some parties.

\textit{Multipartite steering and no-signaling assemblages.-} Since its heuristic conception \cite{S36} and modern reformulation \cite{WJD07} the idea of steering scenario become more and more important quantum phenomenon emphasizing fundamentally non-classical character of nature \cite{R1,R2,steering0,steering1,WJD072, SNC04}. Recently, the notion of steering has been expanded to include also situations described only by no-signaling principles \cite{HS18,SAPHS18,SBCSV15, SHSA20} and consider theories of post-quantum resources \cite{resurceRev}.

%In order to recall the mathematical description of steering experiment, 

Let us consider a typical multipartite scenario of $n+1$ separated parties, where quantum subsystem $C$ is fully trusted (i.e. described by a known Hilbert space $H_C$ under full operational control) and remaining quantum subsystems $A_1,\ldots, A_n$ are untrusted. Assume that such system is described by some joint quantum state and each untrusted party perform local measurements on respective subsystem choosing at random out of $m_i$ measurement settings $x_i=0,\ldots, m_i-1$ and obtaining one of $k_i$ outcomes $a_i=0,\ldots, k_i-1$. This scenario may be then fully described by the triple $n,\mathbf{m},\mathbf{k}$ where $\mathbf{m}=(m_1,\ldots,m_n)$, $\mathbf{k}=(k_1,\ldots,k_n) $ and dimension $d_C$ of  $H_C$. In case like that, possible subnormalized states describing subsystem $C$ conditioned upon choice of measurements $\mathbf{x}_n=(x_1,\ldots,x_n)$ and  outcomes $\mathbf{a}_n=(a_1,\ldots,a_n)$ form a \textit{quantum assemblage} $\Sigma=\left\{\sigma_{\mathbf{a}_n|\mathbf{x}_n}\right\}_{\mathbf{a}_n,\mathbf{x}_n}$ consisting of positive operators (acting on a trusted subsystem $C$) for which there exist a quantum state $\rho_{A_1\ldots A_nC}$ of some composed system $A_1\ldots A_nC$ and elements of local POVM $M^{(A_1)}_{a_1|x_1}, \ldots, M^{(A_n)}_{a_n|x_n}$ such that
$\sigma_{\mathbf{a}_n|\mathbf{x}_n}=\mathrm{Tr}_{A_1\ldots A_n}(M^{(A_1)}_{a_1|x_1}\otimes \ldots M^{(A_n)}_{a_n|x_n}\otimes \mathds{1}_C\rho_{A_1\ldots A_nC}).$

Among assemblages with quantum description we may distinguish specific subclass of \textit{local assemblages} (assemblages with local hidden states model \cite{SAPHS18}) given by those assemblages admitting $
\sigma_{\mathbf{a}_n|\mathbf{x}_n}=\sum_j q_j p^{(1)}_j(a_1|x_1)\ldots p^{(n)}_j(a_n|x_n)\sigma_j
$ where $q_i\geq 0, \sum_j q_j=1$, $\sigma_j$ are some states of trusted subsystem $C$ and $\left\{p^{(i)}_j(a_i|x_i)\right\}_{a_i,x_i}$ denotes conditional probability distributions for untrusted subsystem $A_i$ respectively. %Careful analysis shows that quantum assemblages of this form are simply convex combinations of products of states and appropriate deterministic boxes (correlations). As such they do not provide steering phenomena that go beyond classical stochastic description.

Finally, one can discard the restriction of quantum description of the composed system, demanding that only trusted subsystem obey quantum mechanical characterization, while the joint system $A_1,\ldots A_nC$ is governed by some possibly post-quantum theory (see for example \cite{Witworld}) equipped only with fundamental physical constrains of no-signaling. Probabilistic description of steering experiment in this generalized approach is given by the mathematical structure of \textit{no-signaling assemblage} \cite{SBCSV15,SAPHS18}, i.e. a collection $\Sigma=\left\{\sigma_{\mathbf{a}_n|\mathbf{x}_n}\right\}_{\mathbf{a}_n,\mathbf{x}_n}$ of subnormalised states acting on $d_C$-dimensional Hilbert space of trusted subsystem $C$, such that
\begin{equation}\label{def11}
\forall_{\mathbf{x}_n} \sum_{\mathbf{a}_n} \sigma_{\mathbf{a}_n|\mathbf{x}_n}=\sigma,
\end{equation}where $\sigma$ is some state, and for any set of indexes $I=\left\{i_1,\ldots, i_s\right\}\subset \left\{1,\ldots, n\right\}$ with $1\leq s<n$ there exist an operator $\sigma_{a_{i_1}\ldots a_{i_s}|x_{i_1}\ldots x_{i_s}}$ that fulfills 
\begin{equation}\label{def12}
\forall_{a_k,k\in I}\forall_{\mathbf{x}_n} \sum_{a_j,j\notin I} \sigma_{\mathbf{a}_n|\mathbf{x}_n}=\sigma_{a_{i_1}\ldots a_{i_s}|x_{i_1}\ldots x_{i_s}}.
\end{equation}

%For a particular, fixed scenario $n,\mathbf{m},\mathbf{k}$ with $d_C$ we will denote set of no-signaling assemblages by $\mathbf{nsA}(n,\mathbf{m},\mathbf{k},d_C)$, set of quantum assemblages by $\mathbf{qA}(n,\mathbf{m},\mathbf{k},d_C)$ and set of local assemblages $\mathbf{lA}(n,\mathbf{m},\mathbf{k},d_C)$.

It is easy to see that distinction between quantum and local assemblages certifies existence of quantum entanglement in composed systems. %If assemblage obtained by local measurement performed on some initial state is non-local, then the initial state cannot be separable.
On the other hand, sets of no-signaling and quantum assemblages coincide for  $n=1$ \cite{G89, HJW93} and start to differ if $n>1$ \cite{SBCSV15,SAPHS18}, i.e. there is a possibility of post-quantum steering.

%Nevertheless, each no-signaling assemblage can be described in terms of virtual measurements performed on some Hermitian operator \cite{Lewenstein, SAPHS18}.

%\begin{theorem}\label{W_Hermitian}

%Family of positive opeartors $\Sigma=\left\{\sigma_{\mathbf{a}_n|\mathbf{x}_n}\right\}_{\mathbf{a}_n, \mathbf{x}_n}$ defines a no-signaling assemblage if and only if there exist a Hermitian operator $W\in \bigotimes_{i=1}^{n}B(H_{A_i})\otimes B(H_{C})$ and POVM elements $M^{(i)}_{a_i|x_i}\in B(H_{A_i})$ for which $\sigma_{\mathbf{a}_n|\mathbf{x}_n}=\mathrm{Tr}_{A_1,\ldots, A_n}(M^{(1)}_{a_1|x_1}\otimes\ldots \otimes M^{(n)}_{a_n|x_n}\otimes \mathds{1}_{C}W)$. 
%\end{theorem}

%It worth mentioning that above description is not the only possible characterization of no-signaling assemblages and theirs specific subclasses. In extensive discussion present in \cite{HS18} no-signaling, quantum and local assemblages have been tie respectively with action  of causal, localizable and local channels.
 
Within that framework question regarding possibility of quantum realization of non-local yet extreme points in the set of all no-signaling assemblages become a nontrivial one. This question present important practical challenge. Namely, in a case when adversary attack of eavesdropper is modeled by possible convex decomposition of assemblage playing a role of resource shared by parties performing some cryptographic task, an affirmative answer to the above question certifies security (against post-quantum resources) of the protocol based on quantum description that is non-local.

The analogous question related to sets of correlations (respectively no-signaling, quantum and local) is known to have a negative answer, regardless of the number of parties, measurements and outcomes \cite{RTHHPRL}. However, as it has been recently shown \cite{RBRH20}, the evoked no-go result is no longer true in the case of no-signaling assemblages, even in the simplest nontrivial case.

%i.e. two untrusted subsystems, each with two measurements settings and two possible outcomes (see also short discussion in the following sections). As we will show in the remaining part this observation has tremendous consequences also for the properties of theory describing channel steering in a multiplies case, where all parties obey no-signaling constrains.

\textit{Multipartite channel steering.-} Starting from the intuition laying behind the notion of post-quantum steering it is natural to consider also a generalized scenario of channel steering introduced for bipartite case in \cite{Piani15} and discussed from resource theory perspective in \cite{bus2,bus1}.
Indeed, let us firstly once more consider $n+1$ quantum subsystem $A_1,\ldots, A_n$ and $C$ where only the last subsystem is fully trusted. Subsystem $C$ initially interacts with all subsystems $A_1,\ldots, A_n$ while such interaction is modeled by an action of some quantum channel (i.e. completely positive and trace preserving map). In next step all subsystems become separated and final description of characterized subsystem $C$ is conditioned upon random quantum measurements performed locally by parties at each subsystem $A_i$ with $x_i$ denoting measurements settings and $a_i$ denoting measurements outcomes respectively. Probabilistic description of evolution of subsystem $C$ (with respect to vectors of labels $\mathbf{x}_n=(x_1,\ldots, x_n)$, $\mathbf{a}_n=(a_1,\ldots, a_n)$) is in that case encapsulated by the notion of \textit{quantum channel assemblage}.

\begin{definition}\label{QK-def} Family $\mathcal{L}=\left\{\Lambda_{\mathbf{a}_n|\mathbf{x}_n}\right\}_{\mathbf{a}_n,\mathbf{x}_n}$ consisting of completely positive maps $\Lambda_{\mathbf{a}_n|\mathbf{x}_n}:B(H_C)\rightarrow B(H_{\tilde{C}})$ defines a quantum channel assemblage if for any state $\rho_C\in B(H_C)$ 
\begin{widetext}
\begin{equation}
\Lambda_{\mathbf{a}_n|\mathbf{x}_n}(\rho_C)=\mathrm{Tr}_{A_1,\ldots, A_n}(M^{(1)}_{a_1|x_1}\otimes\ldots \otimes M^{(n)}_{a_n|x_n}\otimes \mathds{1}_{\tilde{C}}(\mathcal{E}(\rho_{A_1,\ldots, A_n}\otimes \rho_C))),
\end{equation}
\end{widetext}
where $\rho_{A_1,\ldots, A_n}\in \otimes^n_i B(H_{A_i})$ is some state, 
$M^{(i)}_{a_i|x_i}$ are some POVM elements acting on subsystem $A_i$, and $\mathcal{E}:\otimes^n_i B(H_{A_i})\otimes  B(H_C)\rightarrow  \otimes^n_iB(H_{A_i})\otimes B(H_{\tilde{C}})$ is some completely positive and trace preserving map (i.e. channel).
\end{definition}

In analogy with assemblages of states, one can in principle relax the requirement of a quantum description  of untrusted subsystems and demand a channel steering scenario with only certain constrains of no-signaling type, when (trusted) quantum subsystem evolves under interaction and later measurements performed on some joint state described by some possibly post-quantum theory. This lead to the natural notion of \textit{no-signaling channel assemblages}.

\begin{definition}\label{NSK-def} Family $\mathcal{L}=\left\{\Lambda_{\mathbf{a}_n|\mathbf{x}_n}\right\}_{\mathbf{a}_n,\mathbf{x}_n}$ consisting of completely positive maps $\Lambda_{\mathbf{a}_n|\mathbf{x}_n}:B(H_C)\rightarrow B(H_{\tilde{C}})$ defines a no-signaling channel assemblage if
\begin{equation}\label{def11o}
\forall_{\mathbf{x}_n} \sum_{\mathbf{a}_n} \Lambda_{\mathbf{a}_n|\mathbf{x}_n}=\Lambda,
\end{equation}where $\Lambda$ is some completely positive and trace preserving map, and for any subset of indexes $I=\left\{i_1,\ldots, i_s\right\}\subset \left\{1,\ldots, n\right\}$ with $1\leq s<n$ there exists a map $\Lambda_{a_{i_1}\ldots a_{i_s}|x_{i_1}\ldots x_{i_s}}$ such that
\begin{equation}
\forall_{a_k,k\in I}\forall_{\mathbf{x}_n} \sum_{a_j,j\notin I} \Lambda_{\mathbf{a}_n|\mathbf{x}_n}=\Lambda_{a_{i_1}\ldots a_{i_s}|x_{i_1}\ldots x_{i_s}}.
\end{equation}
\end{definition}

On the other hand, restriction of quantum description only to the classical cases leads to the definition of \textit{local channel assemblages} (generalizing concept introduced in \cite{Piani15}).

\begin{definition}\label{lhsK-def} Family $\mathcal{L}=\left\{\Lambda_{\mathbf{a}_n|\mathbf{x}_n}\right\}_{\mathbf{a}_n,\mathbf{x}_n}$ consisting of completely positive maps $\Lambda_{\mathbf{a}_n|\mathbf{x}_n}:B(H_C)\rightarrow B(H_{\tilde{C}})$ defines a local channel assemblage if $\Lambda_{\mathbf{a}_n|\mathbf{x}_n}=\sum_j \prod_i^np^{(i)}_j(a_i|x_i)\Lambda_{j}$, where $\left\{p^{(i)}_j(a_i|x_i)\right\}_{a_i,x_i}$ stand for conditional probabilities related to subsystem $A_i$, and $\Lambda_{j}$ are completely positive maps such that $\Lambda=\sum_j \Lambda_{j}$ is a channel.
\end{definition}

For a particular, fixed scenario $n,\mathbf{m},\mathbf{k}$ with $d_C,d_{\tilde{C}}$ we will denote set of no-signaling channel assemblages by $\mathbf{ns\Lambda}(n,\mathbf{m},\mathbf{k},d_C,d_{\tilde{C}})$, set of quantum channel assemblages by $\mathbf{q\Lambda}(n,\mathbf{m},\mathbf{k},d_C,d_{\tilde{C}})$ and set of local channel assemblages $\mathbf{l\Lambda}(n,\mathbf{m},\mathbf{k},d_C,d_{\tilde{C}})$ (we omit $d_{\tilde{C}}$ if $d_C=d_{\tilde{C}}$).

We will show that the difference between quantum and no-signaling (or local) description of multipartite channel steering can be seen as a result of the difference between sets of quantum and no-signaling (or local) assemblages of Choi matrices \cite{Choi,Jamiol} related to maps forming channel assemblages. 
%In particular, this implies that all introduced families of channel assemblages form a convex sets that in general (as long as $n> 1$) do not coincide.

To characterize channels assemblages we recall the well-know notion of a Choi-Jamio\l{}kowski \cite{Choi,Jamiol} isomorphism  $\mathcal{J}:B(B(H_A),B(H_B))\rightarrow B(H_B)\otimes B(H_A)$ given as $\mathcal{J}(\Lambda)=\Lambda \otimes \mathrm{id}_{A}(|\phi_{AA}^{+}\rangle \langle \phi_{AA}^{+}|)$ where $|\phi_{AA}^{+}\rangle\in H_A\otimes H_A$ is a maximally entangled state. 

With this one can show the following result (see Appendix \ref{appA} for detailed calculations).

%The inverse isomorphism $\mathcal{J}^{-1}:B(H_B)\otimes B(H_A)\rightarrow B(B(H_A),B(H_B))$ can be expressed by 
%\begin{equation}\label{izom_choi_jam1}
%(\mathcal{J}^{-1}(Y))(X)=d_A\mathrm{Tr}_{A}((\mathds{1}_B\otimes X^T)Y),
%\end{equation}where $X\in B(H_A)$, $Y\in B(H_B)\otimes B(H_A)$ and $T$ denote transposition on $B(H_A)$.

%Note that operator $\mathcal{J}(\Lambda)$ is called a Choi matrix of $\Lambda$. It is well known that a linear map $\Lambda$ is completely positive if and only if $\mathcal{J}(\Lambda)$ is positive \cite{Choi}. Moreover, $\Lambda$ is a quantum channel if and only if $\mathcal{J}(\Lambda)$ is positive and $\mathrm{Tr}_B(\mathcal{J}(\Lambda))=\frac{\mathds{1}}{d_A}$ \cite{NC00}. This characterization of quantum channels enable us to introduce the following theorem.

\begin{theorem}\label{thm_no_sig_map}
Family of linear maps $\mathcal{L}=\left\{\Lambda_{\mathbf{a}_n|\mathbf{x}_n}\right\}_{\mathbf{a}_n, \mathbf{x}_n}$ given by $\Lambda_{\mathbf{a}_n|\mathbf{x}_n}:B(H_C)\rightarrow B(H_{\tilde{C}})$ defines a no-signaling  (respectively quantum and local) channel assemblage if and only if a family of Choi matrices $\Sigma=\left\{\mathcal{J}(\Lambda_{\mathbf{a}_n|\mathbf{x}_n})\right\}_{\mathbf{a}_n, \mathbf{x}_n}$ defines a no-signaling (respectively quantum and local) assemblage such that $\mathrm{Tr}_{\tilde{C}}(\sum_{\mathbf{a}_n}\mathcal{J}(\Lambda_{\mathbf{a}_n|\mathbf{x}_n}))=\frac{\mathds{1}}{d_{C}}\in B(H_{C})$. Moreover, if the initial family consists of completely positive maps, then $\mathcal{L}$ is a no-signaling channel assemblage if and only if there exist a Hermitian operator $W\in \bigotimes_{i=1}^{n}B(H_{A_i})\otimes B(H_{\tilde{C}})\otimes B(H_{C})$ such that $\mathrm{Tr}_{A_1,\ldots, A_n, \tilde{C}}(W)=\frac{\mathds{1}}{d_{C}}\in B(H_{C})$ and POVM elements $M^{(i)}_{a_i|x_i}\in B(H_{A_i})$ for which $\mathcal{J}(\Lambda_{\mathbf{a}_n|\mathbf{x}_n})=\mathrm{Tr}_{A_1,\ldots, A_n}(M^{(1)}_{a_1|x_1}\otimes\ldots \otimes M^{(n)}_{a_n|x_n}\otimes \mathds{1}_{\tilde{C}C}W)$. Finally, initial family forms a quantum channel assemblage if and only if there exist a state $\rho\in \bigotimes_{i=1}^{n}B(H_{A_i})\otimes B(H_{\tilde{C}})\otimes B(H_{C})$ such that $\mathrm{Tr}_{A_1,\ldots, A_n, \tilde{C}}(\rho)=\frac{\mathds{1}}{d_{C}}\in B(H_{C})$ and POVM elements $M^{(i)}_{a_i|x_i}\in B(H_{A_i})$ for which $\mathcal{J}(\Lambda_{\mathbf{a}_n|\mathbf{x}_n})=\mathrm{Tr}_{A_1,\ldots, A_n}(M^{(1)}_{a_1|x_1}\otimes\ldots \otimes M^{(n)}_{a_n|x_n}\otimes \mathds{1}_{\tilde{C}C}\rho)$, and it forms a local channel assemblage if and only if $\rho$ can be choose as a separable state $\sum_i\lambda_i\rho^{(i)}_{A_1}\otimes \ldots \otimes \rho^{(i)}_{A_n}\otimes \rho^{(i)}_{\tilde{C}C}\in \bigotimes_{i=1}^{n}B(H_{A_i})\otimes B(H_{\tilde{C}})\otimes B(H_{C})$.
\end{theorem}

In particular, above theorems provide that all considered sets of channels assemblages are indeed convex (as linear images of convex sets). Moreover, $\mathbf{ns\Lambda}(n,\mathbf{m},\mathbf{k},d_C,d_{\tilde{C}})$ and $\mathbf{l\Lambda}(n,\mathbf{m},\mathbf{k},d_C,d_{\tilde{C}})$ are compact (as continuous images of compact sets). 

Let us emphasize that characterization of assemblages provided by causal, localizable and local channels (see \cite{HS18}) together with presented theorems leads to yet another description of channel assemblages.

\begin{theorem}\label{last_char}
Family $\mathcal{L}=\left\{\Lambda_{\mathbf{a}_n|\mathbf{x}_n}\right\}_{\mathbf{a}_n, \mathbf{x}_n}$ of completely positive maps $\Lambda_{\mathbf{a}_n|\mathbf{x}_n}:B(H_C)\rightarrow B(H_{\tilde{C}})$ defines a no-signaling (respectively quantum and local) assemblage of channels if and only if there exists a causal (respectively localizable and local) channel $\Phi:\bigotimes_{i=1}^{n}B(H_{X_i})\otimes B(H_Y)\rightarrow \bigotimes_{i=1}^{n}B(H_{A_i})\otimes B(H_{\tilde{C}})\otimes B(H_{C})$ and orthonormal bases $\left\{|a_i\rangle\right\}$ and $\left\{|x_i\rangle\right\}$ of $H_{A_i}$ and $H_{X_i}$ respectively such that $\Lambda_{\mathbf{a}_n|\mathbf{x}_n}=\mathcal{J}^{-1}(\mathrm{Tr}_{A_1\ldots A_n}(\otimes_{i=1}^{n}|a_i\rangle\langle a_i|\Phi(\otimes_{i=1}^{n}|x_i\rangle\langle x_i|\otimes |0_Y\rangle\langle 0_Y|)))$ and $\mathrm{Tr}_{A_1\ldots A_n\tilde{C}}(\Phi(\otimes_{i=1}^{n}|x_i\rangle\langle x_i|\otimes |0_Y\rangle\langle 0_Y|))=\frac{\mathds{1}}{d_C}$.
\end{theorem}

Note that taking one-dimensional output Hilbert space (i.e. $H_{\tilde{C}}=\mathbb{C}$) in the case of Definitions \ref{QK-def}, \ref{NSK-def} and \ref{lhsK-def} one obtains sets consisting of families of maps with outputs in a form of conditional probabilities. To distinguished that case we will say about \textit{channel correlations} and we will introduce a separate notation $\mathcal{P}=\left\{P_{\mathbf{a}_{n}|\mathbf{x}_{n}}\right\}_{\mathbf{a}_{n}, \mathbf{x}_{n}}$. Starting from a recent idea of intermediate set of correlations (i.e. hybrid no-signaling-quantum correlations \cite{BMRH21}) one can provide further refinement. Indeed, $\mathcal{P}=\left\{P_{\mathbf{a}_{n+1}|\mathbf{x}_{n+1}}\right\}_{\mathbf{a}_{n+1}, \mathbf{x}_{n+1}}$ is of a \textit{hybrid form} if there exist a no-signaling channel assemblage $\mathcal{L}=\left\{\Lambda_{\mathbf{a}_n|\mathbf{x}_n}\right\}_{\mathbf{a}_n, \mathbf{x}_n}$ (with arbitrary output dimension) and measurements $M_{a_{n+1}|x_{n+1}}$ such that $P_{\mathbf{a}_{n+1}|\mathbf{x}_{n+1}}(\cdot)=\mathrm{Tr}(M_{a_{n+1}|x_{n+1}}\Lambda_{\mathbf{a}_n|\mathbf{x}_n}(\cdot)$).
\color{black}

%Recently, the idea of intermediate set of correlations situated between those with no-signaling and quantum realization has been introduced [].
%The so called ć are obtain by quantum measurements performed on (in general) no-signaling assemblages. Starting from the above notion of no-signaling channel assemblages one can in a natural way consider also an intermediate set situated between sets of no-signaling and quantum channel correlations. 

\textit{Quantum realization of extreme points.-} Choi-Jamio\l{}kowski isomorphism enable us to identify appropriate classes of channel assemblages with appropriate subclasses of assemblages (acting on a higher dimensional Hilbert spaces), so not all no-signaling channel assemblages admit quantum realization, as long as a given scenario have at least two untrusted subsystems (i.e. $n>1$). In particular, existence of non-local and quantum extreme points within the set of no-signaling channels assemblages is a nontrivial question that can be resolve by the recent solution for assemblages of states \cite{RBRH20}.

For the sake of completeness we will recall \cite{RBRH20} that assemblage of pure states $\Sigma=\left\{p_{\mathbf{a}_n|\mathbf{x}_n}|\psi_{\mathbf{a}_n|\mathbf{x}_n}\rangle \langle \psi_{\mathbf{a}_n|\mathbf{x}_n}|\right\}_{\mathbf{a}_n,\mathbf{x}_n}$  (where each position $\mathbf{a}_n,\mathbf{x}_n$ is occupied by $0$ or positive operator of rank one) is \textit{inflexible} if any other assemblages of the same pure states $\Sigma'=\left\{q_{\mathbf{a}_n|\mathbf{x}_n}|\psi_{\mathbf{a}_n|\mathbf{x}_n}\rangle \langle \psi_{\mathbf{a}_n|\mathbf{x}_n}|\right\}_{\mathbf{a}_n,\mathbf{x}_n}$ such that $q_{\mathbf{a}_n|\mathbf{x}_n}=0$ whenever $p_{\mathbf{a}_n|\mathbf{x}_n}=0$ satisfies $\Sigma'=\Sigma$. It can be seen that inflexible assemblages are extremal points of the set of all no-signaling assemblages and there is a sufficient condition for inflexibility \cite{RBRH20}.

%In fact they provide a class of exposed points in this set. Indeed, starting from an inflexible assemblage $\Sigma=\left\{\sigma_{\mathbf{a}_n|\mathbf{x}_n}\right\}_{\mathbf{a}_n,\mathbf{x}_n}$ we introduce a linear functional 
%\begin{equation}\label{expression}
%F_{\Sigma}(\tilde{\Sigma})=\sum_{\mathbf{a}_n,\mathbf{x}_n}\mathrm{Tr}(\rho_{\mathbf{a}_n|\mathbf{x}_n}\tilde{\sigma}_{\mathbf{a}_n|\mathbf{x}_n})
%\end{equation}where
%\begin{equation}
%\rho_{\mathbf{a}_n|\mathbf{x}_n}=
%\begin{cases}
%0 & \mbox{for $\sigma_{\mathbf{a}_n|\mathbf{x}_n}=0$,} \\
%\frac{\sigma_{\mathbf{a}_n|\mathbf{x}_n}}{\mathrm{Tr}(\sigma_{\mathbf{a}_n|\mathbf{x}_n})} & \mbox{for $\sigma_{\mathbf{a}_n|\mathbf{x}_n}\neq 0 $.}
%\end{cases} 
%\end{equation}Careful analysis shows that $F_{\Sigma}$ attains its maximal value (over the set of no-signaling assemblages) only for initial $\Sigma$.

%As it has been shown in \cite{RBRH20}, for the simplest nontrivial case of two untrusted subsystems (each with a pair of dichotomic measurements), one can formulate the following sufficient condition of inflexibility (hence extremality).  

\begin{theorem}[\cite{RBRH20}]\label{sufficient nontrivial}
Let $\Sigma\in \mathbf{nsA}(2,2,2,d_C)$ be an assemblage of pure states. If there exists $y_1$ and $y_2$ such that each of the following sets 
$\left\{\sigma_{0a_2|y_1x_2}\right\}_{a_2,x_2}$, $\left\{\sigma_{1a_2|y_1x_2}\right\}_{a_2,x_2}$, $\left\{\sigma_{a_10|x_1y_2}\right\}_{a_1,x_1}$ consists of different rank one operator, then $\Sigma$ is inflexible and not local.
\end{theorem}

Based on the above theorem we may provide an example of quantum (yet non-local) channel assemblage that is an extreme point. %in the set of no-signaling channel assemblages.

\begin{example}\label{example_main}
Let consider $\mathcal{L}=\left\{\Lambda_{ab|xy}\right\}_{a,b,x,y}\in \mathbf{q\Lambda}(2,2,2,2)$ where $\Lambda_{ab|xy}:B(H_C)\rightarrow B(H_C)$ admit quantum realization $\Lambda_{ab|xy}(\cdot)=\mathrm{Tr}_{AB}(P_{a|x}\otimes Q_{b|y}\otimes \mathds{1}_C(\mathds{1}_A\otimes\mathcal{E}(|\phi^+_{AB}\rangle \langle \phi^+_{AB}|\otimes \cdot)))$ with two-dimensional Hilbert spaces $H_A, H_B, H_C$, maximally entangled state $|\phi^+_{AB}\rangle \langle \phi^+_{AB}|\in B(H_A)\otimes B(H_B)$, quantum channel $\mathcal{E}(\cdot)=U(\cdot )U^\dagger$ for $U=|0\rangle \langle 0|\otimes \mathds{1}+ |1\rangle \langle 1|\otimes \sigma_x$ and $P_{0|0}=Q_{0|0}=|0\rangle \langle 0|$ and $P_{0|1}=Q_{0|1}=|+\rangle \langle +|$. 
Using Choi-Jamio\l{}kowski isomorphism we obtain
$\sigma_{ab|xy}=\mathcal{J}(\Lambda_{ab|xy})=\mathrm{Tr}_{AB}(P_{a|x}\otimes Q_{b|y}\otimes \mathds{1}_{CC} |\psi_{ABCC}\rangle \langle \psi_{ABCC}|),
$ where $|\psi_{ABCC}\rangle=\frac{1}{2}\left(|0000\rangle+|0011\rangle+|1110\rangle+ |1101\rangle\right)$. Assuming that pairs $a,x$ labels rows and pairs $b,y$ label columns, assemblage of Choi matrices $\Sigma=\left\{\sigma_{ab|xy}\right\}_{a,b,x,y}$ is given in a convenient graphical representation.
\begin{equation}\label{choi_ass}
\Sigma=\frac{1}{4}\begin{pmatrix}
\begin{array}{cc|cc}
 2|\phi\rangle \langle \phi| &  0 & |\phi\rangle \langle \phi| & |\phi\rangle \langle \phi | \\  
 0 & 2|\varphi\rangle \langle \varphi|&|\varphi\rangle \langle \varphi|   &|\varphi\rangle \langle \varphi|\\ \hline
  |\phi\rangle \langle \phi|& |\varphi\rangle \langle \varphi| &|\xi\rangle \langle\xi |&  |\theta\rangle \langle\theta | \\
 |\phi\rangle \langle \phi| &  |\varphi\rangle \langle \varphi|& |\theta\rangle \langle\theta | & |\xi\rangle \langle\xi | 
\end{array}
\end{pmatrix}
\end{equation}with $|\phi\rangle =\frac{1}{\sqrt{2}}(|00\rangle +|11\rangle),|\varphi\rangle =\frac{1}{\sqrt{2}}(|10\rangle +|01\rangle)$ and $|\xi\rangle =\frac{1}{\sqrt{2}}(|\phi\rangle +|\varphi\rangle),|\theta\rangle =\frac{1}{\sqrt{2}}(|\phi\rangle -|\varphi\rangle)$. 

The above assemblage is inflexible and non-local (Theorem \ref{sufficient nontrivial}). Due to the previous discussion initial quantum channel assemblage $\mathcal{L}$ is an extreme point of $\mathbf{ns\Lambda}(2,2,2,2)$ while $\mathcal{L}\notin \mathbf{l\Lambda}(2,2,2,2)$. 
\end{example}

\textit{Multipartie channel steering in an asymmetric scenario.-} Notion of multipartite channel steering can be further modified, when we allow for signaling between certain (but not all) subsystems. For the sake of simplicity, we will consider this asymmetric generalization only in the tripartite case, as this configuration captures all interesting features. Namely, it provide operationally bipartite setting, where subsystem $BC$ are treated as single one (with untrusted part $B$ described by possibly post-quantum theory).

Let us consider subsystem $A$ that will be eventually separated form subsystems $B$ and $C$ with only subsystem $C$ being fully trusted and described by the principles of quantum theory. Similarly to the description introduced in the previous sections, assume that subsystem $C$ initially interacts with subsystem $AB$ and after that subsystems $A$ and $BC$ become separated. Let both untrusted subsystems $A$ and $B$ provide measurements with (chosen at random) settings $x$, $y$ and possible outcomes $a$, $b$ respectively. Probabilistic description of evolution of quantum subsystem $C$ with assumption of no-signaling between $A$ and $BC$ is provided by an \textit{asymmetric channel assemblage}.

\begin{definition}\label{asym-def} Family $\mathcal{L}=\left\{\Lambda_{ab|xy}\right\}_{a,b,x,y}$ consisting of completely positive maps $\Lambda_{ab|xy}:B(H_C)\rightarrow B(H_{\tilde{C}})$ defines an asymmetric no-signaling channel assemblage if i) $\forall_{x,x',y,b} \sum_a \Lambda_{ab|xy}=\sum_a \Lambda_{ab|x'y}$, ii) $
\forall_{y,y',x,a} \sum_b \mathrm{Tr}(\Lambda_{ab|xy})=\sum_b \mathrm{Tr}(\Lambda_{ab|xy'})$, and iii) $
\forall_{x,y} \sum_{a,b} \Lambda_{ab|xy}=\Lambda$ where $\Lambda$ is some fixed quantum channel (i.e. trace preserving map).
\end{definition}

For a particular, fixed scenario $\mathbf{m},\mathbf{k}$ with $d_C,d_{\tilde{C}}$ we will denote set of asymmetric no-signaling channel assemblages by $\mathbf{\tilde{ns}\Lambda}(2,\mathbf{m},\mathbf{k},d_C,d_{\tilde{C}})$. %(as before we will omit $d_{\tilde{C}}$ if $d_C=d_{\tilde{C}}$).

By arguments similar to the proof of Theorem \ref{thm_no_sig_map} we obtain the following characterization.

\begin{theorem}\label{thm_no_sig_channe}
Family of linear maps $\mathcal{L}=\left\{\Lambda_{ab|xy}\right\}_{a,b,x,y}$ given by $\Lambda_{ab|xy}:B(H_C)\rightarrow B(H_{\tilde{C}})$ defines an asymmetric no-signaling channel assemblage if and only if a family of positive Choi matrices $\Sigma=\left\{\sigma_{ab|xy}=\mathcal{J}(\Lambda_{ab|xy})\right\}_{ab|xy}$ fulfill the following conditions i) $\forall_{x,x',y,b} \sum_a \sigma_{ab|xy}=\sum_a \sigma_{ab|x'y}$, ii) $\forall_{y,y',x,a} \sum_b \mathrm{Tr}_{\tilde{C}}(\sigma_{ab|xy})=\sum_b \mathrm{Tr}_{\tilde{C}}(\sigma_{ab|xy'})$, and iii) $
\forall_{x,y} \sum_{a,b} \sigma_{ab|xy}=\sigma$ where $\sigma\in B(H_{\tilde{C}})\otimes B(H_C)$ is some fixed state such that $\mathrm{Tr}_{\tilde{C}}(\sigma)=\frac{\mathds{1}}{d_C}$.
\end{theorem}

Note that under provided assumptions (i.e. Definition \ref{asym-def}) not all examples of channel assemblages that are extremal in multipartite paradigm, remain that way in the asymmetric scenario (however, one can still construct such examples - see Appendix \ref{appB}).

Indeed, consider channel assemblage $\mathcal{L}=\left\{\Lambda_{ab|xy}\right\}_{a,b,x,y}$ form Example \ref{example_main}, now as an element of $\mathbf{\tilde{ns}\Lambda}(2,2,2,2)$. Assume that $\Sigma=\mathcal{J}(\Lambda)$ and that we have a convex decomposition $\Sigma=p\Sigma^{(1)}+(1-p)\Sigma^{(2)}$. Due to conditions i)-iii) in Definition \ref{asym-def}, any $\Sigma^{(i)}$ must be of the form 
\begin{equation}
\Sigma^{(i)}=\frac{1}{4}\begin{pmatrix}
\begin{array}{cc|cc}
 2\alpha|\phi\rangle \langle \phi| &  0 & \eta|\phi\rangle \langle \phi| & \kappa|\phi\rangle \langle \phi | \\  
 0 & 2\delta|\varphi\rangle \langle \varphi|&\eta|\varphi\rangle \langle \varphi|   &\kappa|\varphi\rangle \langle \varphi|\\ \hline
\beta|\phi\rangle \langle \phi|& \epsilon|\varphi\rangle \langle \varphi| &\eta|\xi\rangle \langle\xi |&  \kappa|\theta\rangle \langle\theta | \\
 \gamma|\phi\rangle \langle \phi| &  \zeta|\varphi\rangle \langle \varphi|& \eta|\theta\rangle \langle\theta | & \kappa|\xi\rangle \langle\xi | 
\end{array}
\end{pmatrix}
\end{equation}where in particular $\alpha,\beta,\gamma,\delta,\epsilon,\zeta,\eta,\kappa,\geq 0$, $2\alpha=\eta+\kappa$, $2\delta=\eta+\kappa$ and $\alpha+\delta=2$ (one can consider only single coefficients $\eta,\kappa$ for the whole third and fourth column respectively since both columns consist of four different rank one operators). Note that this leads to $\alpha=\delta=1$.

%and since
%\begin{equation}
%\mathrm{Tr}_{\tilde{C}}(|\phi\rangle \langle \phi|)=\frac{\mathds{1}}{2}=\mathrm{Tr}_{\tilde{C}}(|\varphi\rangle \langle \varphi|)
%\end{equation}
%\begin{equation}
%\mathrm{Tr}_{\tilde{C}}(|\xi\rangle \langle\xi |)=|+\rangle \langle +|
%\end{equation}
%\begin{equation}
%\mathrm{Tr}_{\tilde{C}}(|\theta\rangle \langle\theta |)=|-\rangle \langle -|
%\end{equation}we see that $g=h=1$. 

Nevertheless, when $\beta=\frac{3}{2}, \gamma=\frac{1}{2}, \zeta=\frac{3}{2}, \varepsilon=\frac{1}{2}$ for $\Sigma^{(1)}$ and $\beta=\frac{1}{2}, \gamma=\frac{3}{2}, \zeta=\frac{1}{2}, \epsilon=\frac{3}{2}$ for $\Sigma^{(2)}$ ($\alpha,\delta,\eta,\kappa=1$) we have $\Sigma=\frac{1}{2}\Sigma^{(1)}+\frac{1}{2}\Sigma^{(2)}$ while $\Sigma\neq\Sigma^{(1)}\neq\Sigma^{(2)}$. In the end, $\Sigma$ in not extremal %in the considered sets of no-signaling channel assemblages with relaxed conditions (\ref{thm11}-\ref{thm33}) 
and so is $\Lambda$.

%(i.e. equivalently due to Theorem \ref{thm_no_sig_channe} initial $\Lambda$ is not extremal in the set of all asymmetric channel assemblages).

However, the above reasoning shows that for any $a,b=0,1$ we get
\begin{equation}\label{relation_x}
\Lambda_{ab|00}=\Lambda^{(i)}_{ab|00}
\end{equation}where $\Lambda^{(i)}$ stands for asymmetric channel assemblage related to assemblage of Choi matrices from considered decomposition. %In particular we get $\mathrm{Tr}(M_{c|z}\Lambda_{ab|00}(\rho))=\mathrm{Tr}(M_{c|z}\Lambda^{(i)}_{ab|00}(\rho))$ for any $a,b=0,1$, any any state $\rho$ of qubit subsystem $C$ and any measurement (i.e. POVM elements) $M_{c|z}$ on that subsystem.

Considered paradigm is an example of a post-quantum scenario in which quantum realizable elements (channel assemblages) could provide security against no-signaling eavesdropper in certain cryptographic tasks. Moreover, looking form operational perspective, as classical input and output of subsystem $B$ may be treated as a description of subsystem $C$, this scenario can be seen as a practically bipartite experimental setup, where such security is presented.

Indeed, putting initial state as $\rho=|0\rangle \langle 0|$ and given measurements as $M_{0|0}=|0\rangle \langle 0|,M_{0|0}=|1\rangle \langle 1| $ we obtain correlations (for fixed $x=0$, $y=0$,$z=0$) 
\begin{equation}\label{key_1}
p(000|000)=\mathrm{Tr}(|0\rangle \langle 0|\Lambda_{00|00}(|0\rangle \langle 0|))=\frac{1}{2}
\end{equation}
\begin{equation}\label{key_2}
p(111|000)=\mathrm{Tr}(|1\rangle \langle 1|\Lambda_{11|00}(|0\rangle \langle 0|))=\frac{1}{2}
\end{equation}that seen as a bipartite correlations (in a cut $A|BC$) provide a perfect key. Note that with channel assemblage as a resource sharing between $A$ and $BC$ attack of the eavesdropper may be simulated by a family of completely positive maps $\mathcal{\tilde{L}}=\left\{\Lambda_{abe|xy0}\right\}_{a,b,e,x,y}$ (extension of considered channel assemblage $\mathcal{L}$) such that
i) $\forall_{x,x',y,b,e} \sum_a \Lambda_{abe|xy0}=\sum_a \Lambda_{abe|x'y0}$, ii) $\forall_{y,y',x,a,e} \sum_b \mathrm{Tr}(\Lambda_{abe|xy0})=\sum_b \mathrm{Tr}(\Lambda_{abe|xy'0})$ and $\forall_{x,y} \sum_{a,b,e} \Lambda_{abe|xy0}=\Lambda$ and iv) $\forall_{e} \Lambda_{abe|xy0}=q_e\Lambda^{(e)}_{ab|xy}$ for some $q_e\geq 0$ and no-signaling channel assemblage $\mathcal{L}^{(e)}=\left\{\Lambda^{(e)}_{ab|xy}\right\}_{a,b,x,y}$. In other words such attack is given by a convex combination $\Lambda_{ab|xy}=\sum_{e}q_e\Lambda^{(e)}_{ab|xy}$ where no-signaling channel assemblages $\mathcal{L}^{(e)}$ introduce biased (toward bit $0$ or $1$) versions of correlations (\ref{key_1}) and (\ref{key_2}). However, this is impossible as (\ref{relation_x}) holds and attack (within proposed paradigm) cannot be successful. Observe that (\ref{relation_x}) remains valid even with vi) removed in the above description.%(\ref{def240}).

Similar observation can be made in a different paradigm, when quantum channel behavior (acting from subsystem $C$) is given, while possible attack is simulated by the convex decomposition into \textit{asymmetric hybrid channel behaviors} $\mathcal{P}^{(e)}=\left\{P^{(e)}_{abc|xyz}\right\}_{a,b,c,x,y,z}$, where $P^{(e)}_{abc|xyz}(\cdot)=\mathrm{Tr}_{\tilde{C}}(\tilde{M}^{(e)}_{c|z}\tilde{\Lambda}^{(e)}_{ab|xy}(\cdot))$ with some asymmetric channel assemblage $\mathcal{L}^{(e)}=\left\{\Lambda^{(e)}_{ab|xy}\right\}_{a,b,x,y}$ (acting between subsystem $C$ and arbitrary $\tilde{C}$) and some POVM elements $M^{(e)}_{c|z}$ (acting on subsystem $\tilde{C}$).

Indeed, one can show (see calculations in Appendix \ref{appC}) that for $\mathcal{P}=\left\{P_{abc|xyz}\right\}_{a,b,c,x,y,z}$ with $P_{abc|xyz}(\cdot)=\mathrm{Tr}(M_{c|z}\Lambda_{ab|xy}(\cdot))$ where $\mathcal{L}=\left\{\Lambda_{ab|xy}\right\}_{a,b,x,y}$ is as in Example \ref{example_main} and $M_{0|0}=|0\rangle \langle 0|$, $M_{0|1}=|+\rangle \langle +|$, we have $P_{abc|000}=P^{(e)}_{abc|000}$ for all $a,b,c$ and any $P^{(e)}$ from such convex decomposition. Moreover, $p(abc|000)=P_{abc|000}(|0\rangle \langle 0|)$ fulfill (\ref{key_1}, \ref{key_2}).

%Let $\tilde{\mathcal{L}}=\left\{\tilde{\Lambda}_{ab|xy}\right\}_{a,b,x,y}$ be any asymmetric channel assemblage (acting between quantum subsystems $C$ and $\tilde{C}$) with some POVM elements $\tilde{M}_{c|z}$ such that $\tilde{P}_{abc|xyz}=\mathrm{Tr}(\tilde{M}_{c|z}\tilde{\Lambda}_{ab|xy})$ defines an asymmetric hybrid channel behavior $\tilde{\mathcal{P}}=\left\{\tilde{P}_{abc|xyz}\right\}_{a,b,c,x,y,z}$. Observe that as $\mathcal{J}(\tilde{P}_{abc|xyz})=\mathrm{Tr}_{\tilde{C}}(\tilde{M}_{c|z}\otimes \mathds{1}_C\mathcal{J}(\tilde{\Lambda}_{ab|xy}))$, due to Theorem \ref{thm_no_sig_channe} we obtain a no-signaling assemblage $\tilde{\Sigma}=\left\{\mathcal{J}(\tilde{P}_{abc|xyz})\right\}_{a,b,c,x,y,z}$ fulfilling $\tilde{\sigma}_C=\frac{\mathds{1}}{d_C}$.

\color{black}

It is finally worth noting, that the analogous construction of a perfect key obtained from bipartite quantum correlations with security against adversary equipped with resource in a form of no-signaling correlations remains an open problem.

\textit{Discussion.-} We introduced and analyzed concept of multipartite no-signaling channel assemblages and related subclasses of channel assemblages also in the relaxed framework with no-signaling constrains binding only some parties. %In particular we described considered structures in terms of assemblages of Choi matrices and we discussed the question of quantum realization of extreme (yet non-local) points in set of all no-signaling channel assemblages.

Despite these results, there are still some open questions related to addressed topics. First of all, it would we interesting to provide an expansion of presented results in terms of not only sufficient but sufficient and necessary conditions for extremality among no-signaling channel assemblages (at least in the simplest nontrivial case) as well. Similarly, further characterization of extremality within paradigm of asymmetric channel steering is also needed. Finally, from operational perspective, it also would be important to propose concrete cryptographic protocols based on structure of considered convex sets.

%Note that given results and settings can be further generalized based on the idea of hybrid no-signaling-quantum correlations \cite{BMRH21}. These considerations will be presented elsewhere.

\begin{acknowledgments}
\textit{Acknowledgments}- M.B. and P.H. acknowledge support by the Foundation for Polish Science (IRAP project, ICTQT, contract no. MAB/2018/5, co-financed by EU within Smart Growth Operational Programme). M. B. acknowledges partial support from by DFG (Germany) and NCN (Poland) within the joint funding initiative Beethoven2 (Grant No. 2016/23/G/ST2/04273). R.R. acknowledges support from the Start-up Fund 'Device-Independent Quantum Communication Networks' from The University of Hong Kong, the Seed Fund 'Security of Relativistic Quantum Cryptography' and the Early Career Scheme (ECS) grant 'Device-Independent Random Number Generation and Quantum Key Distribution with Weak Random Seeds'.
\end{acknowledgments}

\appendix
\section{Proof of Theorem \ref{thm_no_sig_map}}\label{appA}

In order to justify characterization given in the main text and provide a proof for Theorem \ref{thm_no_sig_map} let us evoke the following description of no-signaling assemblages \cite{Lewenstein, SAPHS18} and simple auxiliary lemma.

\begin{theorem}[\cite{Lewenstein, SAPHS18}]\label{W_Hermitian}

Family of positive opeartors $\Sigma=\left\{\sigma_{\mathbf{a}_n|\mathbf{x}_n}\right\}_{\mathbf{a}_n, \mathbf{x}_n}$ defines a no-signaling assemblage if and only if there exist a Hermitian operator $W\in \bigotimes_{i=1}^{n}B(H_{A_i})\otimes B(H_{C})$ and POVM elements $M^{(i)}_{a_i|x_i}\in B(H_{A_i})$ for which $\sigma_{\mathbf{a}_n|\mathbf{x}_n}=\mathrm{Tr}_{A_1,\ldots, A_n}(M^{(1)}_{a_1|x_1}\otimes\ldots \otimes M^{(n)}_{a_n|x_n}\otimes \mathds{1}_{C}W)$. 
\end{theorem}

\begin{lem}\label{ext_krauss}Let $\mathcal{E}:B(H_A)\rightarrow B(H_{\tilde{A}})\otimes B(H_B)$ be a completely positive and trace preserving map. There exist some state $\tilde{\rho}_B\in  B(H_B)$ and some completely positive and trace preserving map $\tilde{\mathcal{E}}:B(H_A)\otimes B(H_B)\rightarrow B(H_{\tilde{A}})\otimes B(H_B)$ such that
\begin{equation}\label{cel}
\mathcal{E}(\rho_A)=\tilde{\mathcal{E}}(\rho_A\otimes \tilde{\rho}_B)
\end{equation}for all states $\rho_A\in B(H_A)$.
\end{lem}

Let $\mathcal{L}=\left\{\Lambda_{\mathbf{a}_n|\mathbf{x}_n}\right\}_{\mathbf{a}_n, \mathbf{x}_n}\in \mathbf{ns\Lambda}(n,\mathbf{m},\mathbf{k},d_C,d_{\tilde{C}})$. In that case $\mathcal{J}(\Lambda_{\mathbf{a}_n|\mathbf{x}_n})\in B(H_{\tilde{C}})\otimes B(H_{C})$ are positive operators such that, according to the linearity of $\mathcal{J}$, fulfill relations (\ref{def11}) and (\ref{def12}) in the main text. Moreover, $\mathcal{J}(\sum_{\mathbf{a}_n}\Lambda_{\mathbf{a}_n|\mathbf{x}_n})=\mathcal{J}(\Lambda)$ and $\mathrm{Tr}_{\tilde{C}}(\mathcal{J}(\Lambda))=\frac{\mathds{1}}{d_{C}}$, as $\Lambda$ is completely positive and trace preserving. 

Conversely, assume that $\Sigma=\left\{\mathcal{J}(\Lambda_{\mathbf{a}_n|\mathbf{x}_n})\right\}_{\mathbf{a}_n, \mathbf{x}_n}\in \mathbf{nsA}(n,\mathbf{m},\mathbf{k},d_Cd_{\tilde{C}})$ and $\mathrm{Tr}_{\tilde{C}}(\mathcal{J}(\Lambda))=\frac{\mathds{1}}{d_{C}}$. If so then $\Lambda$ is a completely positive and trace preserving map and according to the linearity of $\mathcal{J}^{-1}$ completely positive maps $\Lambda_{\mathbf{a}_n|\mathbf{x}_n}=\mathcal{J}^{-1}(\mathcal{J}(\Lambda_{\mathbf{a}_n|\mathbf{x}_n}))$ fulfill constraints stated in Definition \ref{NSK-def} in the main text. This proves the first statement regarding no-signaling realization.

The second part is a direct consequence of the above reasoning, characterization of no-signaling assemblages presented in Theorem \ref{W_Hermitian} and the fact that for any no-signaling assemblage $\Sigma=\left\{\sigma_{\mathbf{a}_n|\mathbf{x}_n}\right\}_{\mathbf{a}_n, \mathbf{x}_n}\in \mathbf{nsA}(n,\mathbf{m},\mathbf{k},d_{\tilde{C}}d_C)$ of the form
\begin{equation}
\sigma_{\mathbf{a}_n|\mathbf{x}_n}=\mathrm{Tr}_{A_1,\ldots, A_n}(M^{(1)}_{a_1|x_1}\otimes\ldots \otimes M^{(n)}_{a_n|x_n}\otimes \mathds{1}_{\tilde{C}C}W),
\end{equation}we have $\mathrm{Tr}_{A_1,\ldots, A_n,\tilde{C}}(W)=\mathrm{Tr}_{\tilde{C}}(\sum_{\mathbf{a}_n}\sigma_{\mathbf{a}_n|\mathbf{x}_n})$.
This concludes part of the proof related to no-signaling constrains.

To tackle part regarding quantum realization, assume now that $\mathcal{L}=\left\{\Lambda_{\mathbf{a}_n|\mathbf{x}_n}\right\}_{\mathbf{a}_n, \mathbf{x}_n}\in \mathbf{q\Lambda}(n,\mathbf{m},\mathbf{k},d_C,d_{\tilde{C}})$, then for any state $\rho_C\in B(H_C)$ we have
\begin{widetext}
\begin{equation}
\Lambda_{\mathbf{a}_n|\mathbf{x}_n}(\rho_C)=\mathrm{Tr}_{A_1,\ldots, A_n}(M^{(1)}_{a_1|x_1}\otimes\ldots \otimes M^{(n)}_{a_n|x_n}\otimes \mathds{1}_{\tilde{C}}(\mathcal{E}(\rho_{A_1,\ldots, A_n}\otimes \rho_C))),
\end{equation}
\end{widetext}where $\rho_{A_1,\ldots, A_n}\in \otimes^n_i B(H_{A_i})$ is some fixed state, 
$M^{(i)}_{a_i|x_i}$ are some POVM elements, and $\mathcal{E}:\otimes^n_i B(H_{A_i})\otimes  B(H_C)\rightarrow  \otimes^n_iB(H_{A_i})\otimes B(H_{\tilde{C}})$ is some completely positive and trace preserving map. From this we can write
\begin{equation}
\mathcal{J}(\Lambda_{\mathbf{a}_n|\mathbf{x}_n})=\mathrm{Tr}_{A_1,\ldots, A_n}(M^{(1)}_{a_1|x_1}\otimes\ldots \otimes M^{(n)}_{a_n|x_n}\otimes \mathds{1}_{\tilde{C}C}\rho)
\end{equation}for a state $\rho=\mathcal{E}\otimes \mathrm{id}_{C}(\rho_{A_1,\ldots, A_n}\otimes |\phi^+_{CC}\rangle \langle \phi^+_{CC}|)\in \bigotimes_{i=1}^{n}B(H_{A_i})\otimes B(H_{\tilde{C}})\otimes B(H_{C})$ that provides a quantum realization of $\Sigma=\left\{\mathcal{J}(\Lambda_{\mathbf{a}_n|\mathbf{x}_n})\right\}_{\mathbf{a}_n, \mathbf{x}_n}$. Moreover,
\begin{widetext}
\begin{equation}
\mathrm{Tr}_{\tilde{C}}(\mathcal{J}(\sum_{\mathbf{a}_n}\Lambda_{\mathbf{a}_n|\mathbf{x}_n}))=\mathrm{Tr}_{A_1,\ldots, A_n, \tilde{C}}(\rho)=\mathrm{Tr}_{\tilde{C}}(|\phi^+_{CC}\rangle \langle \phi^+_{CC}|)=\frac{\mathds{1}}{d_{C}}.
\end{equation}
\end{widetext}

Conversely, let $\Sigma=\left\{\mathcal{J}(\Lambda_{\mathbf{a}_n|\mathbf{x}_n})\right\}_{\mathbf{a}_n, \mathbf{x}_n}\in \mathbf{nsA}(n,\mathbf{m},\mathbf{k},d_{\tilde{C}}d_C)$ and $\mathrm{Tr}_{\tilde{C}}(\mathcal{J}(\Lambda))=\frac{\mathds{1}}{d_{C}}$. In this case, there exists a quantum realization
\begin{equation}
\mathcal{J}(\Lambda_{\mathbf{a}_n|\mathbf{x}_n})=\mathrm{Tr}_{A_1,\ldots, A_n}(M^{(1)}_{a_1|x_1}\otimes\ldots \otimes M^{(n)}_{a_n|x_n}\otimes \mathds{1}_{\tilde{C}C}\rho),
\end{equation}with a state $\rho \in \bigotimes_{i=1}^{n}B(H_{A_i})\otimes B(H_{\tilde{C}})\otimes B(H_{C})$ such that $\mathrm{Tr}_{A_1,\ldots, A_n,\tilde{C}}(\rho)=\frac{\mathds{1}}{d_{C}}$. According to that $\mathcal{J}^{-1}(\rho):B(H_C)\rightarrow \bigotimes_{i=1}^{n}B(H_{A_i})\otimes B(H_{\tilde{C}})$ defines a completely positive and trace preserving map. For arbitrary state $\rho_C\in B(H_C)$ we have
\begin{widetext}
\begin{equation}
\Lambda_{\mathbf{a}_n|\mathbf{x}_n}(\rho_C)=\mathcal{J}^{-1}(\mathcal{J}(\Lambda_{\mathbf{a}_n|\mathbf{x}_n}))(\rho_C)=\mathrm{Tr}_{A_1,\ldots, A_n}(M^{(1)}_{a_1|x_1}\otimes\ldots \otimes M^{(n)}_{a_n|x_n}\otimes \mathds{1}_{\tilde{C}}(\mathcal{J}^{-1}(\rho)(\rho_C))).
\end{equation}
\end{widetext}Because of Lemma \ref{ext_krauss} there exist a state $\rho_{A_1,\ldots, A_n}\in \bigotimes_{i=1}^{n}B(H_{A_i})$ and a completely positive and trace preserving map $\mathcal{E}:\bigotimes_{i=1}^{n}B(H_{A_i})\otimes B(H_C)\rightarrow \bigotimes_{i=1}^{n}B(H_{A_i})\otimes B(H_{\tilde{C}})$ such that
\begin{equation}
\mathcal{J}^{-1}(\rho)(\rho_C)=\mathcal{E}(\rho_{A_1,\ldots, A_n}\otimes \rho_C)
\end{equation}for any state $\rho_C\in B(H_C)$. If so, then
\begin{widetext}
\begin{equation}
\Lambda_{\mathbf{a}_n|\mathbf{x}_n}(\rho_C)=\mathrm{Tr}_{A_1,\ldots, A_n}(M^{(1)}_{a_1|x_1}\otimes\ldots \otimes M^{(n)}_{a_n|x_n}\otimes \mathds{1}_{\tilde{C}}(\mathcal{E}(\rho_{A_1,\ldots, A_n}\otimes \rho_C))),
\end{equation}
\end{widetext}which ends the proof of the first part of the thesis regarding quantum channel assemblages.

The second part of the thesis concerning quantum description follows from the argument analogous to the one presented in the no-signaling case.

By a reasoning similar to the one above and the fact that any assemblage with local hidden state model admits quantum realization with measurements performed on some fully separable state we conclude the proof.

\section{Example of quantum extreme point in asymmetric scenario}\label{appB}

Consider a asymmetric qunatum channel assemblage $\Lambda=\left\{\Lambda_{ab|xy}\right\}_{a,b,x,y}$ of qubit maps $\Lambda_{ab|xy}:B(H_C)\rightarrow B(H_C)$ with quantum realization given as
\begin{widetext}
\begin{equation}
\Lambda_{ab|xy}(\cdot)=\mathrm{Tr}_{AB}(P_{a|x}\otimes Q_{b|y}\otimes \mathds{1}_C(\mathds{1}_A\otimes\mathcal{E}_{CNOT}(|\phi^+_{AB}\rangle \langle \phi^+_{AB}|\otimes \cdot)))
\end{equation}
\end{widetext}with $|\phi^+_{AB}\rangle \langle \phi^+_{AB}|$ being a maximally entangled state of two qubits, $\mathcal{E}(\cdot)=U_{CNOT}(\cdot )U_{CNOT}^*$ and $Q_{0|0}=|\theta_1\rangle \langle \theta_1|,\ Q_{1|0}=|\theta_2\rangle \langle \theta_2|$, $P_{0|0}=|\theta_3\rangle \langle \theta_3|,\ P_{1|0}=|\theta_4\rangle \langle \theta_4|$, $Q_{0|1}=P_{0|1}=|+\rangle \langle +|, Q_{1|1}=P_{1|1}=|-\rangle \langle -|$ where
\begin{equation}\nonumber
|\theta_1\rangle=\frac{1}{\sqrt{5}}|0\rangle+\frac{2}{\sqrt{5}}|1\rangle,\ |\theta_2\rangle=\frac{2}{\sqrt{5}}|0\rangle-\frac{1}{\sqrt{5}}|1\rangle,
\end{equation}
\begin{equation}\nonumber
|\theta_3\rangle=\frac{1}{\sqrt{3}}|0\rangle+\sqrt{\frac{2}{3}}|1\rangle,\ |\theta_4\rangle=\sqrt{\frac{2}{3}}|0\rangle-\frac{1}{\sqrt{3}}|1\rangle.
\end{equation}
Define related assemblages of Choi matrices by $\mathcal{J}(\Lambda_{ab|xy})=\Lambda_{ab|xy}\otimes \mathrm{id}_C(|\phi^+_{CC}\rangle \langle \phi^+_{CC}|)=\mathrm{Tr}_{AB}(P_{a|x}\otimes Q_{b|y}\otimes \mathds{1}_{CC} |\psi_{ABCC}\rangle \langle \psi_{ABCC}|)$ with
\begin{equation}
|\psi_{ABCC}\rangle=\frac{1}{2}\left(|0000\rangle+|0011\rangle+|1110\rangle+ |1101\rangle\right).
\end{equation}Explicit form of this assemblage $\Sigma=\left\{\mathcal{J}(\Lambda_{ab|xy})\right\}_{a,b,x,y}$ of Choi matrices is given as
\begin{equation}\nonumber
\Sigma=\begin{pmatrix}
\begin{array}{cc|cc}
 |\phi_1\rangle \langle \phi_1| &   |\phi_2\rangle \langle \phi_2| &  |\phi_5\rangle \langle \phi_5| &  |\phi_6\rangle \langle \phi_6| \\  
  |\phi_3\rangle \langle \phi_3| &  |\phi_4\rangle \langle \phi_4|& |\phi_7\rangle \langle \phi_7|   & |\phi_8\rangle \langle \phi_8|\\ \hline
 |\phi_9\rangle \langle \phi_9| & |\phi_{10}\rangle \langle \phi_{10}| &|\phi_{13}\rangle \langle \phi_{13}|&  |\phi_{14}\rangle \langle \phi_{14}|\\
|\phi_{11}\rangle \langle \phi_{11}| & |\phi_{12}\rangle \langle \phi_{12}|& |\phi_{15}\rangle \langle \phi_{15}|& |\phi_{16}\rangle \langle \phi_{16}|
\end{array}
\end{pmatrix}
\end{equation}where
\begin{equation}\nonumber
|\phi_{1}\rangle=\frac{1}{2\sqrt{15}}\left(|00\rangle+|11\rangle+2\sqrt{2}|10\rangle+2\sqrt{2}|01\rangle\right),
\end{equation}
\begin{equation}\nonumber
|\phi_{2}\rangle=\frac{1}{2\sqrt{15}}\left(2|00\rangle+2|11\rangle-\sqrt{2}|10\rangle-\sqrt{2}|01\rangle\right),
\end{equation}
\begin{equation}\nonumber
|\phi_{3}\rangle=\frac{1}{2\sqrt{15}}\left(\sqrt{2}|00\rangle+\sqrt{2}|11\rangle-2|10\rangle-2|01\rangle\right),
\end{equation}
\begin{equation}\nonumber
|\phi_{4}\rangle=\frac{1}{2\sqrt{15}}\left(2\sqrt{2}|00\rangle+2\sqrt{2}|11\rangle+|10\rangle+|01\rangle\right),
\end{equation}
\begin{equation}\nonumber
|\phi_{5}\rangle=\frac{1}{2\sqrt{6}}\left(|00\rangle+|11\rangle+\sqrt{2}|10\rangle+\sqrt{2}|01\rangle\right),
\end{equation}
\begin{equation}\nonumber
|\phi_{6}\rangle=\frac{1}{2\sqrt{6}}\left(\sqrt{2}|00\rangle+\sqrt{2}|11\rangle-|10\rangle-|01\rangle\right),
\end{equation}
\begin{equation}\nonumber
|\phi_{7}\rangle=\frac{1}{2\sqrt{6}}\left(\sqrt{2}|00\rangle+\sqrt{2}|11\rangle+\sqrt{2}|10\rangle+\sqrt{2}|01\rangle\right),
\end{equation}
\begin{equation}\nonumber
|\phi_{8}\rangle=\frac{1}{2\sqrt{6}}\left(\sqrt{2}|00\rangle+\sqrt{2}|11\rangle+|10\rangle+|01\rangle\right),
\end{equation}
\begin{equation}\nonumber
|\phi_{9}\rangle=\frac{1}{2\sqrt{10}}\left(|00\rangle+|11\rangle+2|10\rangle+2|01\rangle\right),
\end{equation}
\begin{equation}\nonumber
|\phi_{10}\rangle=\frac{1}{2\sqrt{10}}\left(2|00\rangle+2|11\rangle-|10\rangle-|01\rangle\right),
\end{equation}
\begin{equation}\nonumber
|\phi_{11}\rangle=\frac{1}{2\sqrt{10}}\left(|00\rangle+|11\rangle-2|10\rangle-2|01\rangle\right),
\end{equation}
\begin{equation}\nonumber
|\phi_{12}\rangle=\frac{1}{2\sqrt{10}}\left(2|00\rangle+2|11\rangle+|10\rangle+|01\rangle\right),
\end{equation}
\begin{equation}\nonumber
|\phi_{13}\rangle=\frac{1}{2}|+\rangle|+\rangle,\ |\phi_{14}\rangle=\frac{1}{2}|-\rangle|-\rangle,
\end{equation}
\begin{equation}\nonumber
|\phi_{15}\rangle=\frac{1}{2}|-\rangle|-\rangle,\ |\phi_{16}\rangle=\frac{1}{2}|+\rangle|+\rangle.
\end{equation}

Assume that there is a convex decompostion $\Sigma=p\Sigma_1+(1-p)\Sigma_2$, where both $\Sigma_i$ are assemblages described by conditions i)-iii) of Theorem \ref{thm_no_sig_channe} in the main text. If so then 
\begin{equation}\nonumber
\Sigma_i=\begin{pmatrix}
\begin{array}{cc|cc}
 \alpha|\phi_1\rangle \langle \phi_1| &  \beta |\phi_2\rangle \langle \phi_2| & \gamma |\phi_5\rangle \langle \phi_5| &  \delta|\phi_6\rangle \langle \phi_6| \\  
  \alpha|\phi_3\rangle \langle \phi_3| &  \beta|\phi_4\rangle \langle \phi_4|& \gamma|\phi_7\rangle \langle \phi_7|   & \delta|\phi_8\rangle \langle \phi_8|\\ \hline
 \alpha|\phi_9\rangle \langle \phi_9| & \beta|\phi_{10}\rangle \langle \phi_{10}| &\gamma|\phi_{13}\rangle \langle \phi_{13}|&  \delta|\phi_{14}\rangle \langle \phi_{14}|\\
\alpha|\phi_{11}\rangle \langle \phi_{11}| & \beta|\phi_{12}\rangle \langle \phi_{12}|& \gamma|\phi_{15}\rangle \langle \phi_{15}|& \delta|\phi_{16}\rangle \langle \phi_{16}|
\end{array}
\end{pmatrix}
\end{equation}for some non-negative coefficients $\alpha,\beta,\gamma,\delta$ (since each initial column consists of different rank one operators and condition i) of Theorem \ref{thm_no_sig_channe} in the main text must be fulfilled).

Observe that due to condition ii) of Theorem \ref{thm_no_sig_channe} in particular we have
\begin{widetext}
\begin{equation}
\alpha\mathrm{Tr}_{C}(|\phi_1\rangle \langle \phi_1|)+\beta\mathrm{Tr}_{C}(|\phi_2\rangle \langle \phi_2|)=\gamma\mathrm{Tr}_{C}(|\phi_5\rangle \langle \phi_5|)+\delta\mathrm{Tr}_{C}(|\phi_6\rangle \langle \phi_6|)
\end{equation}and
\begin{equation}
\alpha\mathrm{Tr}_{C}(|\phi_9\rangle \langle \phi_9|)+\beta\mathrm{Tr}_{C}(|\phi_{10}\rangle \langle \phi_{10}|)=\gamma\mathrm{Tr}_{C}(|\phi_{13}\rangle \langle \phi_{13}|)+\delta\mathrm{Tr}_{C}(|\phi_{14}\rangle \langle \phi_{14}|).
\end{equation}
\end{widetext}These equalities implies that coefficients $\alpha,\beta,\gamma,\delta$ must fulfill
\begin{equation}\label{system_eq}
\begin{cases}
\frac{9}{5}\alpha+\frac{6}{5}\beta-\frac{3}{2}\gamma-\frac{3}{2}\delta=0\\
4\alpha-4\beta-5\gamma+5\delta=0\\
\alpha+\beta-\gamma-\delta=0
\end{cases}.
\end{equation}According to Kronecker-Capelli theorem space of solutions of linear system (\ref{system_eq}) is one dimensional. As $\alpha=\beta=\gamma=\delta=1$ form a non-zero solution of (\ref{system_eq}), normalization constrain implies that $\alpha=\beta=\gamma=\delta=1$ is the only possible choice of coefficients, so $\Sigma=\Sigma_i$ for $i=1,2$ and due to Theorem \ref{thm_no_sig_channe} in the main text considered quantum channel assemblages is an extreme point in the set of all asymmetric no-signaling channel assemblages $\mathbf{\tilde{ns}\Lambda}(2,2,2,2)$.

\section{Proof that $P_{abc|000}=P^{i}_{abc|000}$}\label{appC}
Let $\tilde{\mathcal{L}}=\left\{\tilde{\Lambda}_{ab|xy}\right\}_{a,b,x,y}$ be any asymmetric channel assemblage (acting between quantum subsystems $C$ and $\tilde{C}$) with some POVM elements $\tilde{M}_{c|z}$ such that the following formula
\begin{equation}
\tilde{P}_{abc|xyz}(.)=\mathrm{Tr}(\tilde{M}_{c|z}\tilde{\Lambda}_{ab|xy}(.))
\end{equation}defines an asymmetric hybrid channel behavior $\tilde{\mathcal{P}}=\left\{\tilde{P}_{abc|xyz}\right\}_{a,b,c,x,y,z}$. Observe that as $\mathcal{J}(\tilde{P}_{abc|xyz})=\mathrm{Tr}_{\tilde{C}}(\tilde{M}_{c|z}\otimes \mathds{1}_C\mathcal{J}(\tilde{\Lambda}_{ab|xy}))$, due to Theorem \ref{thm_no_sig_channe} in the main text we obtain a no-signaling assemblage $\tilde{\Sigma}=\left\{\sum_c\mathcal{J}(\tilde{P}_{abc|xyz})\right\}_{a,b,x,y}$ fulfilling $\sum_{a,b,c}\mathcal{J}(\tilde{P}_{abc|xyz}))=\frac{\mathds{1}}{d_C}$. Moreover, for any fixed $b,y$ the following sub-collection of operators $\tilde{\Sigma}^{(b|y)}=\left\{\mathcal{J}(\tilde{P}_{abc|xyz})\right\}_{a,c,x,z}$ form a no-signaling assemblage up to normalization (i.e. $\mathrm{Tr}(\sum_{ac}\mathcal{J}(\tilde{P}_{abc|xyz}))$ may be lesser than $1$).

Recall the quantum channel assemblage $\mathcal{L}=\left\{\Lambda_{ab|xy}\right\}_{a,b,x,y}$ from Example \ref{example_main} in the main text and related quantum channel behavior $\mathcal{P}=\left\{P_{abc|xyz}\right\}_{a,b,c,x,y,z}$ defined by $P_{abc|xyz}(\cdot)=\mathrm{Tr}(M_{c|z}\Lambda_{ab|xy}(\cdot))$ with $M_{0|0}=|0\rangle \langle 0|$ and $M_{0|1}=|+\rangle \langle +|$ (where $a,b,c,x,y,z\in\left\{0,1\right\}$). Observe that 
\begin{equation}\label{ex_as_1}
\Sigma^{(0|0)}=\frac{1}{8}\begin{pmatrix}
\begin{array}{cc|cc}
 2|0\rangle \langle 0| &  2|1 \rangle\langle 1| &2|+\rangle \langle+ | &2|-\rangle \langle - | \\  
 0 & 0&0  &0\\ \hline
  |0\rangle \langle 0|& |1\rangle \langle 1| &|+\rangle \langle+ |& |-\rangle \langle- | \\
 |0\rangle \langle 0| & |1\rangle \langle 1|& |+\rangle \langle +| & |-\rangle \langle -| 
\end{array}
\end{pmatrix}
\end{equation}
\begin{equation}\label{ex_as_2}
\Sigma^{(1|0)}=\frac{1}{8}\begin{pmatrix}
\begin{array}{cc|cc}
0 & 0&0  &0\\
 2|1\rangle \langle 1| &  2|0\rangle\langle 0| &2|+\rangle \langle+ | &2|-\rangle \langle - | \\  \hline
  |1\rangle \langle 1|& |0\rangle \langle 0| &|+\rangle \langle+ |& |-\rangle \langle- | \\
 |1\rangle \langle 1| & |0\rangle \langle 0|& |+\rangle \langle +| & |-\rangle \langle -| 
\end{array}
\end{pmatrix}
\end{equation}
\begin{equation}\label{ex_as_3}
\Sigma^{(0|1)}=\frac{1}{8}\begin{pmatrix}
\begin{array}{cc|cc}
 |0\rangle \langle 0| &  |1\rangle \langle 1| &|+\rangle \langle +| &|-\rangle \langle -| \\  
 |1\rangle \langle 1| & |0 \rangle\langle 0|&|+\rangle \langle +|   &|-\rangle \langle -|\\ \hline
  |+\rangle \langle +|& |+\rangle \langle +| &2|+\rangle \langle+ |& 0 \\
 |-\rangle \langle -| & |-\rangle \langle -|& 0 & 2|-\rangle \langle- | 
\end{array}
\end{pmatrix}
\end{equation}and 
\begin{equation}\label{ex_as_4}
\Sigma^{(1|1)}=\frac{1}{8}\begin{pmatrix}
\begin{array}{cc|cc}
 |0\rangle \langle 0| &  |1\rangle \langle 1| &|+\rangle \langle +| &|-\rangle \langle -| \\  
 |1\rangle \langle 1| & |0 \rangle\langle 0|&|+\rangle \langle +|   &|-\rangle \langle -|\\ \hline
  |-\rangle \langle -|& |-\rangle \langle -| &0& 2|-\rangle \langle- | \\
 |+\rangle \langle +| & |+\rangle \langle +|& 2|+\rangle \langle+ | & 0
\end{array}
\end{pmatrix}
\end{equation}where $\Sigma^{(b|y)}=\left\{\sigma_{ac|xz}^{(b|y)}\right\}_{a,c,x,z}=\left\{\mathcal{J}(P_{abc|xyz})\right\}_{a,c,x,z}$.

Assume that there exists a convex combination $\mathcal{P}=p_1\mathcal{P}_1+p_2\mathcal{P}_2$ such that both $\mathcal{P}_i=\left\{P^{(i)}_{abc|xyz}\right\}_{a,b,c,x,y,z}$ are asymmetric hybrid channel behaviors. Note that all $\Sigma_i^{(b|y)}=\left\{\sigma^{(b|y),i}_{ac|xz}\right\}_{a,c,x,z}=\left\{\mathcal{J}(P^{(i)}_{abc|xyz})\right\}_{a,c,x,z}$ are no-signaling assemblages up to normalization (due to the discussion in the above paragraph). Without the loss of generality fix $i$. Assumption concerning convex combination together with explicit form of (\ref{ex_as_1}-\ref{ex_as_4}) imply that $\sigma^{(b|y),i}_{ac|xz}=\alpha^{(b|y),i}_{ac|xz}\sigma^{(b|y)}_{ac|xz}$ for any $a,b,c,x,y,z$ and some non-negative coefficients $\alpha^{(b|y),i}_{ac|xz}$. Theorem \ref{sufficient nontrivial} in the main text applied respectively to $\Sigma^{(0|1)}$ and $\Sigma^{(1|1)}$ provides that in fact $\sigma^{(b|1),i}_{ab|xy}=\alpha^{(b|1),i}\sigma^{(c|z)}_{ab|xy}$ for any $a,b,c,x,z$ and some non-negative coefficients $\alpha^{(b|1),i}$. To simplify notation we put $\alpha=\alpha^{(0|1),i}$, $\beta=\alpha^{(1|1),i}$. Moreover, analysis of $\Sigma^{(0|0)}$ and $\Sigma^{(1|0)}$ shows that for all $c,z$ we have $\alpha^{(0|0),i}_{0c|0z}=\gamma_1$, $\alpha^{(1|0),i}_{1c|0z}=\gamma_2$, $\alpha^{(0|0),i}_{0c|1z}=\delta_1$, $\alpha^{(0|0),i}_{1c|1z}=\delta_3$, $\alpha^{(1|0),i}_{0c|1z}=\delta_2$, and $\alpha^{(1|0),i}_{1c|1z}=\delta_4$ for some $\gamma_1,\gamma_2, \delta_1,\delta_2,\delta_3,\delta_4$.

With this knowledge let us finally consider the no-signaling assemblage $\Sigma_i=\left\{\sum_c\mathcal{J}(P^{(i)}_{abc|xyz})\right\}_{a,b,x,y}$ given explicitly as (compare with (\ref{ex_as_1}-\ref{ex_as_4}))

\begin{equation}\label{final_ass}
\Sigma_i=\frac{1}{8}\begin{pmatrix}
\begin{array}{cc|cc}
 2\gamma_1\mathds{1} &  0 &\alpha\mathds{1} &\beta\mathds{1}\\  
 0 & 2\gamma_2\mathds{1}&\alpha\mathds{1}  &\beta\mathds{1}\\ \hline
  \delta_1\mathds{1}& \delta_2\mathds{1} &2\alpha|+\rangle \langle+ |& 2\beta|-\rangle \langle- | \\
 \delta_3\mathds{1}& \delta_4\mathds{1}&2\alpha |-\rangle \langle -| & 2\beta|+\rangle \langle +| 
\end{array}
\end{pmatrix}.
\end{equation}In particular, the following no-signaling equality 
\begin{equation}
(\delta_1+\delta_2)\mathds{1}=2\alpha|+\rangle \langle +|+2\beta|-\rangle \langle -|
\end{equation}can be true if and only if $\alpha=\beta$. Since no-signaling conditions also imply $2\gamma_1=\alpha+\beta$ and $2\gamma_2=\alpha+\beta$ while $2\alpha+2\beta=1$ (normalization) we get $\gamma_1=\gamma_2=1$. This concludes the argument as by applying inverse of the Choi-Jamio\l{}kowski isomorphism we obtain $P_{abc|000}=P^{i}_{abc|000}$ for all $a,b,c$.


\begin{thebibliography}{99}
\bibitem{Lewenstein}A. Ac{\'i}n, R. Augusiak, D. Cavalcanti, C. Hadley, J. K. Korbicz, M. Lewenstein, Ll. Masanes, M. Piani, \textit{Unified Framework for Correlations in Terms of Local Quantum Observables}, Phys. Rev. Lett. \textbf{104}, 140404 (2010).
%\bibitem{AS2016}N. Ananth and M. Senthilvelan, Int. J. Theor. Phys. 55, 1854 (2016).
%\bibitem{BHR21}M. Banacki, P. Horodecki, R. Ramanathan, \textit{In preparation}.
\bibitem{BMRH21}M. Banacki, P. Mironowicz, R. Ramanathan, P. Horodecki, \textit{Hybrid no-signaling-quantum correlations}, New J. Phys. \textbf{24}, 083003 (2022).
\bibitem{BRH21}M. Banacki, R. Ravell Rodr{\'i}guez, P. Horodecki, \textit{Edge of the set of no-signaling assemblages}, Phys. Rev. A \textbf{103}, 052434 (2021).
\bibitem{Bell64} J. S. Bell, \textit{On the Einstein Podolsky Rosen paradox}, Physics \textbf{1}, 195 (1964).
\bibitem{WJD073}J. Bowles, T. V{\'e}rtesi, M. T. Quintino, N. Brunner, \textit{One-way Einstein-Podolsky-Rosen Steering}, Phys. Rev. Lett. \textbf{112}, 200402 (2014).
\bibitem{Witworld}P. J. Cavalcanti, J. H. Selby, J. Sikora, T. D. Galley, A. B. Sainz, \textit{Witworld: A generalised probabilistic theory featuring post-quantum steering}, arXiv:2102.06581 (2021). 
\bibitem{R1}D. Cavalcanti, P. Skrzypczyk, \textit{Quantum steering: a review with focus on semidefinite programming
}, Rep. Prog. Phys. \textbf{80}, 024001 (2017).
\bibitem{steering0}D. Cavalcanti, P. Skrzypczyk, G. Aguilar, R. V. Nery, P.H. Souto Ribeiro, S. P. Walborn, \textit{Detection of entanglement in asymmetric quantum networks and multipartite quantum steering}, Nat Commun \textbf{6}, 7941 (2015).
\bibitem{R2}R. Uola, A. C.S. Costa, H. C. Nguyen, O. G{\"u}hne, \textit{Quantum steering}, Rev. Mod. Phys. \textbf{92}, 015001 (2020).
\bibitem{resurceRev}E. Chitambar, G. Gour, \textit{Quantum resource theories}, Rev. Mod. Phys. \textbf{91}, 025001 (2019).
\bibitem{steering1}J.-L. Chen, C. Ren, C. Chen, X.-J. Ye, A. K. Pati, \textit{Bell's Nonlocality Can be Detected by the Violation of Einstein-Podolsky-Rosen Steering Inequality}, Sci Rep \textbf{6}, 39063 (2016).
\bibitem{Choi}M.-D. Choi, \textit{Completely positive linear maps on complex matrices}, Linear Alg. Appl. \textbf{10}, 285, (1795).
\bibitem{Witness}D. Chru{\'s}ci{\'n}ski, G. Sarbicki, \textit{Entanglement witnesses: construction, analysis and classification}, J. Phys. A: Math. Theor. \textbf{47}, 483001 (2014).
\bibitem{EPR} A. Einstein, B. Podolsky, N. Rosen, \textit{Can Quantum-Mechanics Description of Physical Reality be Considered Complete?}, Phys Rev \textbf{47}, 777 (1935).
\bibitem{G89} N. Gisin, \textit{Stochastic quantum dynamics and relativity}, Helvetica Physica Acta \textbf{62}, 363 (1989).
\bibitem{NC00} M. A. Nielsen, I. L. Chuang, \textit{Quantum computation and quantum information}, Cambridge (2000).
\bibitem{PR}S. Popescu, D. Rohrlich,\textit{Quantum nonlocality as an axiom}, Found. Phys. \textbf{24}, 379-385 (1994).
\bibitem{S36}E. Schr\"{o}dinger, \textit{Probability relations between separated systems}, \textit{Mathematical Proceedings of the Cambridge Philosophical Society} \textbf{32}, 446-452 (1936).
\bibitem{Jamiol}A. Jamio\l{}kowski, \textit{Linear transformations which preserve trace and positive semidefiniteness of operators}, Rep. Math. Phys. \textbf{3}, 275 (1972).
\bibitem{HS18}M. J. Hoban, A. B. Sainz, \textit{A channel-based framework for steering, non-locality and beyond}, New J. Phys. \textbf{20}, 053048 (2018).
\bibitem{4H}R. Horodecki, P. Horodecki, M. Horodecki, K. Horodecki, \textit{Quantum entanglement}, Rev. Mod. Phys. \textbf{81}, 865 (2009).
\bibitem{HJW93} L. P. Hughston, R. Jozsa, K. Wooters, \textit{A complete classification of quantum ensembles having a given density matrix},  Phys. Lett. A \textbf{183}, 14 (1993).
\bibitem{Piani15}M. Piani, \textit{Channel steering}, Journal of the Optical Society of America B Vol. 32, \textbf{4}, A1-A7 (2015).
\bibitem{WJD072}M. T. Quintino, T. V{\'e}rtesi, D. Cavalcanti, R. Augusiak, M. Demianowicz, A. Acín, N. Brunner, \textit{Inequivalence of entanglement, steering, and Bell nonlocality for general measurements}, Phys. Rev. A \textbf{92}, 032107 (2015).
\bibitem{RHB21} R. Ramanathan, M. Banacki P. Horodecki, \textit{No-signaling-proof randomness extraction from public weak sources}, arXiv:2108.08819 (2021).
\bibitem{RBRH20} R. Ramanathan, M. Banacki, R. Ravell Rodríguez, P. Horodecki, \textit{Single trusted qubit is necessary and sufficient for quantum realisation of extremal no-signaling correlations}, arXiv:2004.14782 (2020).
\bibitem{RTHHPRL}R. Ramanathan, J. Tuziemski, M. Horodecki and P. Horodecki, \textit{No quantum realization of extremal no-signaling boxes}, Phys. Rev. Lett. \textbf{117}, 050401 (2016).
\bibitem{bus2}D. Rosset, D. Schmid, F. Buscemi, \textit{Type-Independent Characterization of Spacelike Separated Resources}, Phys. Rev. Lett. \textbf{125}, 210402 (2020).
\bibitem{SAPHS18}A. B. Sainz, L. Aolita, M. Piani, M. J. Hoban, P. Skrzypczyk, \textit{A formalism for steering with local quantum measurements}, New J. Phys. \textbf{20}, 083040 (2018). 
\bibitem{SBCSV15} A. B. Sainz, N. Brunner, D. Cavalcanti, P. Skrzypczyk, T. V\'{e}rtesi, \textit{Postquantum Steering}, Phys. Rev. Lett. \textbf{115}, 190403 (2015).
\bibitem{SHSA20}A. B. Sainz, M. J. Hoban, P. Skrzypczyk, L. Aolita, \textit{Bipartite Postquantum Steering in Generalized Scenarios}, Phys. Rev. Lett. \textbf{125}, 050404 (2020).
\bibitem{scaranibook}V. Scarani, \textit{Bell Nonlocality}, Oxford University Press, Oxford (2019).
\bibitem{bus1}D. Schmid, D. Rosset, F. Buscemi, \textit{The type-independent resource theory of local operations and shared randomness}, Quantum \textbf{4}, 262 (2020).
\bibitem{SNC04}P. Skrzypczyk, M. Navascu\'{e}s, D. Cavalcanti, \textit{Quantifying Einstein-Podolsky-Rosen Steering}, Phys. Rev. Lett. \textbf{112}, 180404 (2014).
\bibitem{WJD07} H. M. Wiseman, S. J. Jones, A. C. Doherty, \textit{Steering, Entanglement, Nonlocality, and the Einstein-Podolsky-Rosen Paradox}, Phys. Rev. Lett. \textbf{98} 140402 (2007).
%\bibitem{Supp}Supplementary Material
\end{thebibliography}
\end{document}